\documentclass{elsarticle}

\usepackage{amsmath}
\usepackage{amsthm}
\usepackage{amssymb}

\journal{Algorithmica}

\newtheorem{theorem}{Theorem}
\newtheorem{corollary}{Corollary}
\newtheorem{definition}{Definition}
\newtheorem{lemma}{Lemma}
\newenvironment{example}{\paragraph{Example}}{\hfill \qed \\}

\DeclareMathOperator{\polylog}{polylog}

\begin{document}

\begin{frontmatter}

\title{Optimal Encodings for Range Majority Queries \tnoteref{CPM}}

\tnotetext[CPM]{An early version of this article appeared in {\em Proc. CPM
2014} \cite{NT14}.}

\author[GN]{Gonzalo Navarro\fnref{Nucleo}}
\ead{gnavarro@dcc.uchile.cl}
\author[ST]{Sharma V. Thankachan}
\ead{sharma.thankachan@gmail.com}

\fntext[Nucleo]{Partially funded by Millennium Nucleus Information and
Coordination in Networks ICM/FIC P10-024F, Chile.}

\address[GN]{Department of Computer Science, University of Chile,Chile.}
\address[ST]{Georgia Institute of Technology, USA.}

\begin{abstract}
We study the problem of designing a data structure that reports the 
positions of the distinct
$\tau$-majorities within any range of an array $A[1,n]$, without storing $A$. 
A $\tau$-majority in a range $A[i,j]$, for $0<\tau< 1$, is an element that 
occurs more than $\tau(j-i+1)$ times in $A[i,j]$. We show that 
$\Omega(n\log(1/\tau))$ bits are necessary for any data structure able just
to count the number of distinct $\tau$-majorities in any range. Then, we 
design a structure using $O(n\log(1/\tau))$ bits that returns one position of
each $\tau$-majority of $A[i,j]$ in $O((1/\tau)\log\log_w(1/\tau)\log n)$ time,
on a RAM machine with word size $w$ (it can output any further position where 
each $\tau$-majority occurs in $O(1)$ additional time). Finally, we show how
to remove a $\log n$ factor from the time by adding $O(n\log\log n)$ bits of
space to the structure. 
\end{abstract}

\end{frontmatter}

\section{Introduction}

Given an array $A[1, n]$ of $n$ arbitrary elements, an \emph{array range query} problem asks us to build a data structure over $A$,
 such that whenever a range $[i, j]$ with $1 \leq i \leq j \leq n$ arrives as an input, we can efficiently answer queries on the elements in $A[i, j]$
\cite{skala}. 
Many array range queries arise naturally as subproblems of combinatorial
problems, and are also of direct interest in data mining applications.
Well-known examples are range minimum queries (RMQs, which seek the smallest 
element in $A[i, j]$)~\cite{BV93}, top-$k$ queries (which report the $k$ 
largest elements in $A[i, j]$)~\cite{bro}, range selection queries (which report
the $k$th largest element in $A[i,j]$) \cite{CW13},
and colored top-$k$ queries (which
report the $k$ largest distinct elements in $A[i,j]$)~\cite{NNN}.
 
An {\em encoding} for array range queries is a data structure that answers the
queries without accessing $A$. This is useful when the values of $A$ are not 
of interest themselves, and thus $A$ may be deleted, potentially saving a lot 
of space. It is also useful when array $A$ does not fit in main memory, so it
can be kept in secondary storage while a much smaller encoding can be
maintained in main memory, speeding up queries. In this 
setting, instead of reporting an element in $A$, we only report a position in 
$A$ containing the element. Otherwise, in many cases we would be able to 
reconstruct $A$ via queries on the encodings, and thus these could not be small
(e.g., $A[i]$ would be the only answer to the range query $A[i,i]$ for all the 
example queries given above).
As examples of encodings, RMQs can be solved in constant time using just 
$2n+o(n)$ bits \cite{RMQ1} and, using $O(n\log k)$ bits, top-$k$ queries can be 
solved in $O(k)$ time \cite{grossi} and range selection queries
in $O(\log k / \log\log n)$ time \cite{NRR14}.

Frequency-based array range queries, in particular variants of
heavy-hitter-like problems, are very popular in data mining. Queries
such as finding the most frequent element in a range (known as the range mode 
query) are known to be harder than problems like RMQs.
For range mode queries, known data structures with constant query 
time require nearly quadratic space~\cite{R2}. The best known linear-space 
solution requires $O(\sqrt{n/\log n})$ query time~\cite{stacs}, and
conditional lower bounds given in that paper show that a significant improvement is 
highly unlikely.

Still, efficient solutions exist for some useful variants of the
range mode problem. An example are approximate range mode queries, where we are 
required to output an element whose number of occurrences in $A[i, j]$ is 
at least $1/(1+\epsilon)$ times the number of occurrences of the mode in $A[i,
j]$~\cite{lbo,lboo}. 

In this paper we focus on a popular
variant of range mode queries called \emph{range $\tau$-majority queries},
which ask to report any element that occurs more than
$\tau(j-i+1)$ times in $A[i,j]$. 
A version of the problem useful for encodings can be stated as follows (other
variants are possible).

\begin{definition}
Given an array $A[1,n]$, a {\em range $\tau$-majority query} receives a
range $[i,j]$ and returns one position in the range where each
$\tau$-majority in $A[i,j]$ occurs. A {\em $\tau$-majority} is any element that 
occurs more than $\tau(j-i+1)$ times in $A[i,j]$. When $\tau=1/2$ we simply
call it a {\em majority}.
\end{definition}

Range majority queries can be answered in constant time by maintaining a
linear space (i.e., $O(n)$-word or $O(n\log n)$-bit) data 
structure~\cite{Steph}. Similarly, range $\tau$-majority
queries can be solved in time $O(1/\tau)$ and linear space if $\tau$ is fixed
at construction time, or $O(n\log\log n)$ space (i.e., $O(n\log n\log\log n)$
bits) if $\tau$ is given at query time \cite{wads}.
 
In this paper, we focus for the first time on {\em encodings for range
$\tau$-majority queries}.
In this scenario, a valid question is how much space is necessary for an
encoding that correctly answers such queries (we recall that $A$ itself
is not available at query time). We answer that question in
Section~\ref{sec:lb}, proving a lower bound for any encoding that solves even
a weaker query.

\begin{theorem} \label{thm:lb}
Given a real number $0<\tau<1$, any encoding able to count the number of
range $\tau$-majorities in any range $A[i,j]$
must use $\Omega(n\log(1/\tau))$ bits.
\end{theorem}

Since when using $O(n\log n)$ bits we have sufficient space to store $A[1,n]$%
\footnote{Or an equivalent array where each element is replaced by an
identifier in $[1,n]$.}
(and achieve the optimal $O(1/\tau)$ time \cite{wads}), encodings for range
$\tau$-majorities are asymptotically interesting only for $\log(1/\tau) = o(\log n)$.

In Section~\ref{sec:nlglgn} we show how range
$\tau$-majority queries can be solved using $O((n/\tau)\log\log n)$ bits of 
space and $O((1/\tau)\log n)$ query time. 
In Section~\ref{sec:nbits} we reduce the
space to the optimal $O(n\log(1/\tau))$ bits and slightly increase the time.
After spending this time, the structure can report {\em any} of the
positions of any majority in optimal time (e.g., the leftmost
position of each $\tau$-majority in a negligible $O(1/\tau)$ time).
In Section~\ref{sec:constr} we show how to build our structure in $O(n\log n)$ time. 
All the results hold on the RAM model with word size $w=\Omega(\log n)$ bits.

\begin{theorem} \label{thm:main}
Given a real number $0<\tau<1$, there exists an encoding using the optimal
$O(n\log(1/\tau))$ bits that answers range $\tau'$-majority queries, for any
$\tau \le \tau' < 1$, in time $O((1/\tau)\log\log_w(1/\tau)\log n)$, where
$w=\Omega(\log n)$ is the RAM word size in bits. 
It can report any $occ$ further
occurrence positions of the majorities in $O(occ)$ time.
The encoding can be built in $O(n\log n)$ time. 
\end{theorem}

We note that the query time is simply $O((1/\tau)\log n)$ for polylogarithmic
values of $1/\tau$. We also note that the time depends on $\tau$, not $\tau'$.
In Section~\ref{sec:constr} we also show how to obtain a query time that is a 
function of $\tau'$, yet using $O(n\log^2(1/\tau))$ 
bits of space. 

Finally, in Section~\ref{sec:optimal} we derive a new variant
that may use more space but slashes the $\log n$ term from the time complexity.

\begin{theorem} \label{thm:optimal}
Given a real number $0<\tau<1$, there exists an encoding using 
$O(n\log(1/\tau)+n\log\log n)$ bits that answers range $\tau'$-majority queries,
for any $\tau \le \tau' < 1$, in time $O((1/\tau)\log\log_w(1/\tau))$, where
$w=\Omega(\log n)$ is the RAM word size in bits. 
It can report any $occ$ further
occurrence positions of the majorities in $O(occ)$ time.
The encoding can be built in $O(n\log n)$ time.
\end{theorem}

By combining the results of Theorems~\ref{thm:main} and \ref{thm:optimal}, we
obtain the combinations given in Table~\ref{tab:final}.

\begin{table}[t]
\begin{center}
\begin{tabular}{c|c|c}
Condition        &   Space (bits)             & Query time \\ 
& & \\ [-1.5ex]
\hline\\ [-1.5ex]
$1/\tau = \omega(\polylog n)$ & ~~$O(n\log(1/\tau))$~* &
$O((1/\tau)\log\log_w(1/\tau))$ \\ 
& & \\[-1.5ex]
$1/\tau = \Theta(\polylog n)$ & ~~$O(n\log(1/\tau))$~* & $O(1/\tau)~*$ \\
& & \\[-1.5ex]
$1/\tau = o(\polylog n)$ & ~~$O(n\log(1/\tau))$~* & $O((1/\tau)\log n)$ \\
& & \\[-1.5ex]
$1/\tau = o(\polylog n)$ & $O(n\log\log n)$ & $O(1/\tau)~*$ \\
\end{tabular}
\end{center}
\caption{Space-time tradeoffs achieved. We mark the optimal spaces and
times with a *.}
\label{tab:final}
\end{table}

\section{Related Work}

In this section we first cover the state of the art for answering range 
$\tau$-majority
queries. Then, we survey a few results on bitmap representation, and
give a new result that will be useful for this paper. Again, all these results
hold on the RAM model with word size $w=\Omega(\log n)$ bits.

\subsection{Range Majorities}

Range $\tau$-majority queries were introduced by Karpinski and
Nekrich~\cite{KarpinskiN08}, who presented an 
 $O(n/\tau)$-words structure with $O((1/\tau)(\log\log n)^2)$ query time.
Durocher et al.~\cite{Steph} improved their word-space and query time to
$O(n\log(1/\tau))$ and $O(1/\tau)$, respectively. 
 Gagie et al.~\cite{spire} presented another trade-off, where the space is
$O(n(H+1))$ {\em bits} and the query time is $O((1/\tau)\log\log n)$. 
 Here $H\le\lg n$ denotes the empirical entropy of the distribution of
elements in $A$ (we use $\lg$ to denote the logarithm in base 2).
  The best current result in general 
is by Belazzougui et al.~\cite{wads}, where the
space is $O(n)$ words and the query time is $O(1/\tau)$. 
  All these results assume that $\tau$ is fixed at  construction time. 

For the case where
 $\tau$ is also a part of the query input, data structures of space (in words)
$O(n(H+1))$ and $O(n\log n)$ were proposed by Gagie et al.~\cite{spire} and
Chan et al.~\cite{swat}, respectively. 
 Very recently, Belazzougui et al.~\cite{wads} brought down the space
occupancy to $O(n\log\log\sigma)$ words, where $\sigma$ is the number of distinct elements in $A$. 
 The query time is $O(1/\tau)$ in all cases.
Belazzougui et al.~\cite{wads} also presented a compressed solution using
$nH+o(n\log \sigma)$ bits,
  with slightly higher query time. 
All these solutions
include a (sometimes compressed) representation of $A$, thus they are not
encodings. As far as we know, ours is the first encoding for this problem. 
  
For further reading, we recommend the recent survey by Skala~\cite{skala}. 

\subsection{Bitmap Representations}
\label{sec:bitmaps}

Given a bitmap $B[1,m]$ with $n$ 1s, the operation $rank(B,i)$ returns the number
of 1s in $B[1,i]$, whereas operation $select(B,j)$ gives the position of the 
$j$th 1 in $B$. Both operations can be solved in constant time by storing 
$o(m)$ bits in addition to $B$ \cite{Mun96,Cla96}. When $n$ is significantly
smaller than $m$, another useful representation \cite{RRR07} compresses $B$ to
$n\lg\frac{m}{n}+O(n)+o(m)$ bits and retains constant time for both
operations.

When $n$ is much smaller than $m$, even the $o(m)$ extra bits of that
compressed representation \cite{RRR07} are troublesome, and an Elias-Fano-based
\cite{Fan71,Eli74}
compressed representation \cite{OS07} is useful. It requires $n\lg\frac{m}{n}
+ O(n)$ bits, solves $select$ in $O(1)$ time and $rank$ in
$O(\log\frac{m}{n})$ time. The representation considers the positions of all
the 1s in $B$, $x_i = select(B,i)$, and encodes the lowest $b = 
\lceil \lg\frac{m}{n} \rceil$ bits of each $x_i$ in an array $L[1,n]$, 
$L[i] = x_i~\mathrm{mod}~2^b$. Then it defines a bitmap $H[1,2n]$ that encodes 
the highest bits of the $x_i$ values: all the bits at positions 
$i+(x_i ~\mathrm{div}~ 2^b)$ are set in $H$. Bitmap $H$ is indexed for 
constant-time $rank$ and $select$ queries \cite{Mun96,Cla96}. The space 
for $L[1,n]$ is $n\lceil \lg\frac{m}{n} \rceil$ and $H$ uses $2n+o(n)$ bits. 

Now, $select(B,j) = 2^b (select(H,j)-j) + L[i]$ can be computed in constant time.
For $rank(B,i)$, we observe that the $h$th 0 in $H$ represents the point where
the position $B[2^b h]$ is reached in the process of setting the 1s at 
positions $i+(x_i ~\mathrm{div}~ 2^b)$, that is, $x_{i-1} < 2^b h \le x_i$. 
The number of 1s in $H$ up to that 
position is $rank(B,2^b h)$. Therefore, if we write $i = 2^b h + l$, then 
$rank(B,i)$ is between $j_1 = rank(H,select_0(H,h))+1$ and $j_2 = 
rank(H,select_0(H,h+1))$, where $select_0(H,h)$ gives the position of the
$h$th 0 in $H$ and is also computed in constant time and $o(n)$ bits
\cite{Mun96,Cla96}.
Now we binary search for $l$ in $L[j_1,j_2]$, which is increasing in that
range. The range is of length at most $2^b$, so the search takes
$O(b)=O(\log\frac{m}{n})$ time. The final position $j$ returned by the search
is $rank(B,i)$.

The time can be improved to $O(\log\log_w\frac{m}{n}+\log s)$ on a RAM machine 
of $w$ bits by sampling, for each increasing interval of $L$ of length more 
than $s$, one value out of $s$. Predecessor data structures are built on the
samples of each interval, taking at most $O((n/s)\log\frac{m}{n})$ bits. Then 
we first run a predecessor query on $L[j_1,j_2]$, which takes time
$O(\log\log_w\frac{m}{n})$ \cite{PT08}, and finish with an $O(\log s)$-time
binary search between the resulting samples.

\begin{lemma} \label{lem:bitmap}
A bitmap $B[1,m]$ with $n$ 1s can be stored in $n\log\frac{m}{n} +
O((n/s)\log\frac{m}{n} + n)$
bits, so that $select$ queries take $O(1)$ time and $rank$ queries take
$O(\log\log_w\frac{m}{n}+\log s)$, for any $s$, on a RAM machine of $w$ bits.
\end{lemma}

\section{Lower Bounds} \label{sec:lb}

We derive a lower bound on the minimum size range $\tau$-majority encodings
may have, even if we just ask them to count the number of distinct 
$\tau$-majorities present in any range. The idea is to show that we can 
encode a certain combinatorial
object in the array $A$, so that the object can be recovered via range
$\tau$-majority queries. Therefore, in the worst case, the number of bits
needed to solve such queries must be at least the logarithm of the number 
of distinct combinatorial objects that can be encoded.

Consider a sequence of $m$ permutations on $[3k]$. There are $(3k)!^m$ such 
sequences, thus any encoding for them must use at least $m\lg ((3k)!)$ bits 
in the worst case. Now consider the following encoding. Array $A$ will have 
length $n=36\cdot k\cdot m$. To encode the $i$th permutation, $\pi_i = 
(x_1~x_2~\ldots~x_{3k})$, we will write 9 chunks on $A[36k(i-1)+1,36ki]$:

\begin{eqnarray*}
 1,2,3,\ldots,k, & {-}1,{-}2,{-}3,\ldots,{-}2k, & x_1,x_2,x_3,\ldots,x_k \\
 k{+}1,k{+}2,k{+}3,\ldots,2k,& {-}1,{-}2,{-}3,\ldots,{-}2k,& x_1,x_2,x_3,\ldots,x_k \\
 2k{+}1,2k{+}2,2k{+}3,\ldots,3k,& {-}1,{-}2,{-}3,\ldots,{-}2k,& x_1,x_2,x_3,\ldots,x_k \\
 1,2,3,\ldots,k,& {-}1,{-}2,{-}3,\ldots,{-}2k,& x_{k{+}1},x_{k{+}2},x_{k{+}3},\ldots,x_{2k} \\
 k{+}1,k{+}2,k{+}3,\ldots,2k,& {-}1,{-}2,{-}3,\ldots,{-}2k,& x_{k{+}1},x_{k{+}2},x_{k{+}3},\ldots,x_{2k} \\
 2k{+}1,2k{+}2,2k{+}3, \ldots,3k,& {-}1,{-}2,{-}3,\ldots,{-}2k,& x_{k{+}1},x_{k{+}2},x_{k{+}3},\ldots,x_{2k} \\
 1,2,3,\ldots,k,& {-}1,{-}2,{-}3,\ldots,{-}2k,& x_{2k{+}1},x_{2k{+}2},x_{2k{+}3},\ldots,x_{3k} \\
 k{+}1,k{+}2,k{+}3,\ldots,2k,& {-}1,{-}2,{-}3,\ldots,{-}2k,& x_{2k{+}1},x_{2k{+}2},x_{2k{+}3},\ldots,x_{3k} \\
 2k{+}1,2k{+}2,2k{+}3,\ldots,3k,& {-}1,{-}2,{-}3,\ldots,{-}2k,& x_{2k{+}1},x_{2k{+}2},x_{2k{+}3},\ldots,x_{3k} \\
\end{eqnarray*}

We will set $\tau = 1/(2k+2)$ and perform $\tau$-majority queries on parts of 
$A$ to recover any permutation.

Let us start obtaining $\pi_i(1)=x_1$. Let $C[1,36] = A[36k(i-1)+1,36ki]$.
Consider an interval of the form $$C[\ell,3k+g] = \ell,\ell+1,\ldots,k,
{-}1,{-2},\ldots,{-}2k,x_1,x_2,\ldots,x_g,$$ for $1 \le \ell,g \le k$.
Note that $x_1,\ldots,x_g$ are the only values that may appear twice in 
$C[\ell,3k+g]$, precisely, if they belong to $\{\ell,\ldots,k\}$. Note that
elements appearing once in $C[\ell,3k+g]$ are not $\tau$-majorities, since 
$1 \le \tau(3k+g-\ell+1)$ for any values $k,\ell,g$. On the other hand, if 
an element appears twice in $C[\ell,3k+g]$, then it is a $\tau$-majority, since
$2 > \tau(3k+g-\ell+1)$ for any values $k,\ell,g$.

With this tool, we can discover $x_1$ as follows. First, we ask whether there 
is a $\tau$-majority in $C[1,3k+1]$. If there is none, then $x_1 \not\in 
\{1,\ldots,k\}$, and we have to look for it elsewhere (in $C[4k+1,8k]$ or 
$C[8k+1,12k]$). Assume there is a $\tau$-majority in $C[1,3k+1]$; then 
$x_1 \in \{1,\ldots,k\}$. Now we query the range $C[2,3k+1]$. If there is no 
$\tau$-majority, then $x \not\in \{2,\ldots,k\}$, and we conclude that $x_1=1$.
If there is, then $x \in \{2,\ldots,k\}$ and we query the range $C[3,3k+1]$.
If there is no $\tau$-majority, then $x \not\in \{3,\ldots,k\}$ and we 
conclude that $x_1=2$, and so on. The process is continued, if necessary, until
querying the range $C[k,3k+1]$. If, instead, we had originally found out that 
$x \not\in \{1,\ldots,k\}$, then we look for
it analogously in $C[4k+1,8k]$ or $C[8k+1,12k]$.

To look for $x_2$, we consider similarly ranges of the form $C[\ell,3k+2]$, with
identical reasoning. This time, it is possible that element $x_1$ is also 
counted as an answer, but since we already know the value of $x_1$, we simply
subtract 1 from the count in any range $C[\ell,3k+2]$ with $\ell \le x_1$.
This process continues
analogously until we identify $x_k$. The other two thirds of $\pi_i$ are 
extracted analogously from $C[12k+1,24k]$ and $C[24k+1,36k]$.

\begin{example}
Consider encoding $m=1$ permutation $\pi = (1~5~3~9~2~4~6~8~7)$, of size $3k=9$.
Then we set $\tau=1/8$ and the array $A[1,108]$ is as follows:
\begin{eqnarray*}
 1,2,3, & {-}1,{-}2,{-}3,{-4},{-}5,{-}6,& 1,5,3 \\
 4,5,6,& {-}1,{-}2,{-}3,{-4},{-}5,{-}6,& 1,5,3 \\
 7,8,9,& {-}1,{-}2,{-}3,{-4},{-}5,{-}6,& 1,5,3 \\
 1,2,3,& {-}1,{-}2,{-}3,{-4},{-}5,{-}6,& 9,2,4 \\
 4,5,6,& {-}1,{-}2,{-}3,{-4},{-}5,{-}6,& 9,2,4 \\
 7,8,9,& {-}1,{-}2,{-}3,{-4},{-}5,{-}6,& 9,2,4 \\
 1,2,3,& {-}1,{-}2,{-}3,{-4},{-}5,{-}6,& 6,8,7 \\
 4,5,6,& {-}1,{-}2,{-}3,{-4},{-}5,{-}6,& 6,8,7 \\
 7,8,9,& {-}1,{-}2,{-}3,{-4},{-}5,{-}6,& 6,8,7 \\
\end{eqnarray*}

Now we will find $x_1$ (which is 1, but we do not know it yet).
Since $A[1,10]$ has a $\tau$-majority, we know that $x_1 \in \{ 1,2,3 \}$.
Since $A[2,10]$ has no $\tau$-majority, we know that $x_2 \not\in \{ 2,3\}$,
thus we learn $x_1=1$.

Now let us find $x_2$.
Since $A[1,11]$ has one $\tau$-majority, which we know corresponds to $x_1=1$,
we conclude that $x_2 \not\in \{ 1,2,3 \}$. Thus we will have to find it
analogously in $A[13,24]$ or in $A[25,36]$.

Now let us find $x_3$.
Since $A[1,12]$ has two $\tau$-majorities, one of which we know corresponds to 
$x_1=1$, and the other we know does not correspond to $x_2$, we conclude that
$x_3 \in \{ 1,2,3\}$.
Now $A[2,12]$ has one $\tau$-majority. We know it does not correspond to
$x_1=1$ (as it falls outside the range) nor to $x_2$ (as it is not in this
subset). Then it follows that $x_3 \in \{ 2,3\}$. 
Finally, since $A[3,12]$ still has one $\tau$-majority, we conclude 
$x_3 \in \{ 3 \}$, thus $x_3=3$. 

Element $x_2$ will be found in $A[13,24]$. Elements $x_4,x_5,x_6$ will be
obtained from $A[37,72]$ and elements $x_7,x_8,x_9$ from $A[73,108]$.
\end{example}

Now, since $n=36km$ and $\tau=1/(2k+2)$, we have
that any encoding able to answer the above queries requires at least
$$m\lg \left((3k)!\right) > m\left(3k \lg (3k) - 3k \lg e + 1\right) > 
 \frac{n}{12}\left(\lg\left(\frac{3}{2}\frac{1}{\tau}-3\right) - \lg e\right)$$ bits%
\footnote{Bounding $\lg (3k)!$ with integrals one obtains
$3k \lg(3k/e) + 1 \le \lg (3k)! \le (3k+1)\lg((3k+1)/e)+1$.}. 
This is $\Omega(n\log(1/\tau))$ unless $1/\tau \le 2+\frac{2}{3}e = O(1)$, 
thus it is sufficient that we show that $\Omega(n)$ is a lower bound for any 
constant $\tau \ge 1/(2+\frac{2}{3}e) > 1/4$. 

To show that $\Omega(n)$ bits are necessary for any $\tau \ge 1/4$, consider
encoding a bitmap $B[1,m]$ in an array $A[1,4m]$ so that, if $B[i]=0$,
then $A[4(i-1)+1]=1$, $A[4(i-1)+2]=2$, $A[4(i-1)i+3]=3$, and
$A[4i]=4$. Instead, if $B[i]=1$, then $A[4(i-1)+1,4i]=1$. Then, for any
$\tau \ge 1/4$, there is a $\tau$-majority in $A[4(i-1)+1,4i]$ iff $B[i]=1$.
As there are $2^m$ possible bitmaps $B$ and our array is of length $n=4m$,
we need at least $m = n/4 = \Omega(n)$ bits for any encoding. 
Then the proof of Theorem~\ref{thm:lb} is complete.

%
%

\section{An $O((n/\tau)\log\log n)$ Bits Encoding for Range $\tau$-Majorities}
\label{sec:nlglgn}

In this section we obtain an encoding using $O((n/\tau)\log\log n)$ bits and 
solving $\tau$-majority queries in $O((1/\tau)\log n)$ time. In the next 
section we improve the space usage. We assume that $\tau$ is fixed at
construction time. At query time, we will be able to solve any
$\tau'$-majority query for any $\tau \le \tau' < 1$.

\subsection{The Basic Idea}

Consider each distinct symbol $x$ appearing in $A[1,n]$. Now consider the set
of all the segments $S_x$ within $[1,n]$ where $x$ is a $\tau$-majority (this includes,
in particular, all the segments $[k,k]$ where $A[k]=x$). Segments in $S_x$ may
overlap each other. 
Now let $A_x[1,n]$ be a bitmap such that $A_x[k]=1$ iff
position $k$ belongs to some segment in $S_x$. 
We define a second bitmap related to $x$, $M_x$, so that if $A_x[k]=1$, then
$M_x[rank(A_x,k)]=1$ iff $A[k]=x$, where operation $rank$ was defined in
Section~\ref{sec:bitmaps}. 

\begin{example} Let our running example array be
$A[1,7] = \langle 1~3~2~3~3~1~1\rangle$, and
$\tau=1/2$. Then we have the segments $S_x$:
\begin{eqnarray*}
S_1 &=& \{ [1,1], [6,6], [7,7], [6,7], [5,7] \},\\
S_2 &=& \{ [3,3] \},\\
S_3 &=& \{ [2,2], [4,4], [5,5], [4,5], [2,4], [3,5], [4,6], [2,5], [1,5],
[2,6] \},
\end{eqnarray*}
and the corresponding bitmaps $A_x$:
$$
A_1 ~=~ \langle 1~0~0~0~1~1~1 \rangle, ~~~
A_2 ~=~ \langle 0~0~1~0~0~0~0 \rangle, ~~~
A_3 ~=~ \langle 1~1~1~1~1~1~0 \rangle.
$$
Finally, the corresponding bitmaps $M_x$ are:
$$
M_1 ~=~ \langle 1~0~1~1 \rangle, ~~~
M_2 ~=~ \langle 1 \rangle, ~~~
M_3 ~=~ \langle 0~1~0~1~1~0 \rangle.
$$
\vspace*{-5mm}
\end{example}

Then, the following result is not difficult to prove.

\begin{lemma} \label{lem:Ax}
An element $x$ is a $\tau'$-majority in $A[i,j]$ iff $A_x[k]=1$ for all 
$i\le k\le j$, and 1 is a $\tau'$-majority in $M_x[rank(A_x,i),rank(A_x,j)]$.
\end{lemma}
\begin{proof}
If $x$ is a $\tau'$-majority in $A[i,j]$, then it is also a $\tau$-majority.
Thus, by definition, $[i,j] \in S_x$, and therefore all the
positions $k \in [i,j]$ are set to 1 in $A_x$. Therefore, the whole segment
$A_x[i,j]$ is mapped bijectively to $M_x[rank(A_x,i),rank(A_x,j)]$, which is
of the same length. Finally, the number of occurrences of $x$ in $A[i,j]$ is 
the number of occurrences of 1 in $M_x[rank(A_x,i),$ $rank(A_x,j)]$, which 
establishes the result. 

Conversely, if $A_x[k]=1$ for all $i\le k\le j$, then $A[i,j]$ is bijectively
mapped to $M_x[rank(A_x,i),rank(A_x,j)]$, and the 1s in this range correspond
one to one with occurrences of $x$ in $A[i,j]$. Therefore, if 1 is a 
$\tau'$-majority in $M_x[rank(A_x,i),$ $rank(A_x,j)]$, then $x$ is a 
$\tau'$-majority in $A[i,j]$.
\end{proof}

\begin{example}
Value 1 is a majority in $A[5,7]$, and it holds that $A_1[5,7] = \langle
1~1~1 \rangle$ and $M_1[rank(A_1,5),rank(A_1,7)] = M_1[2,4] = \langle 0~1~1
\rangle$, where 1 is a majority.
\end{example}

Thus, with $A_x$ and $M_x$ we can determine whether $x$ is
a majority in a range.

\begin{lemma} \label{lem:algox}
It is sufficient to have $rank$-enabled bitmaps $A_x$ and $M_x$ to determine, 
in constant time, whether $x$ is a $\tau'$-majority in any $A[i,j]$.
\end{lemma}
\begin{proof}
We use Lemma~\ref{lem:Ax}. We compute $i'=rank(A_x,i)$ and $j'=rank(A_x,j)$. 
If $j'-i' \not= j-i$, then $A_x[k]=0$ for some $i \le k \le j$ and thus $x$
is not a $\tau$-majority in $A[i,j]$, hence it is also not a $\tau'$-majority. 
Otherwise, we find out whether 1 is a $\tau'$-majority
in $M_x[i',j']$, by checking whether $rank(M_x,j')-rank(M_x,i'-1)>\tau'(j'-i'+1)$. 
\end{proof}

To find any position $i \le k \le j$ where $A[k]=x$, we need the
operation $select(B,j)$, defined in Section~\ref{sec:bitmaps}.
Then, for example, if $x$ is a $\tau'$-majority in $A[i,j]$, its leftmost 
occurrence in $A[i,j]$ is $i-i'+select(M_x,rank(M_x,i'-1)+1)$. In general,
for any $1 \le t \le rank(M_x,j')-rank(M_x,i'-1)$, we can retrieve the $t$th 
occurrence with $i-i'+select(M_x,rank(M_x,i'-1)+t)$.

\subsection{Coalescing the Bitmaps}

We cannot afford to store (and probe!) all the bitmaps $A_x$ and $M_x$ for all $x$,
however. The next lemma is the first step to reduce the total space to 
slightly superlinear.

\begin{lemma} \label{lem:five1s}
For any position $A[k]=x$ there are at most $2\lceil 1/\tau\rceil$ 1s in $A_x$.
\end{lemma}
\begin{proof}
Consider a process where we start with $A[k]=\,\perp$ for all $k$, and set the 
values $A[k]=x$ progressively. We will distinguish three kinds of changes.

\paragraph{{(1) New segments around $A[k]$ are created in $S_x$}}

Setting $A[k]=x$ creates in $S_x$ all the segments of the form
$[k-k_l,k+k_r]$ for $1>\tau(k_r+k_l+1)$, or $k_l+k_r < 1/\tau-1$. Their
union is the area
$A_x[k-\lceil 1/\tau\rceil+2,\ldots, k+\lceil 1/\tau\rceil-2]=1$, which may
increase the number of 1s in $A_x$ by up to $2\lceil 1/\tau\rceil-3$.

\paragraph{{(2) Segments already covering $A[k]$ are extended}}

Any maximal segment $[l,r] \in S_x$ covering $A_x[k]$ contains $c > \tau(r-l+1)$
occurrences of $x$, but it holds that $c \le \tau(r-l+2)$, otherwise there would
also exist segments $[l-1,r]$ and $[l,r+1]$ in $S_x$, and $[l,r]$ would not be
maximal. Therefore, adding one more occurrence, $A[k]=1$, we get
$c+1 \le \tau(r-l+2+1/\tau)$ occurrences in $[l,r]$. Now it holds that $x$
may be a $\tau$-majority in segments $[l-k_l,r+k_r]$ for all $0 \le k_l+k_r <
1+1/\tau$ (i.e., where $c+1 > \tau(r-l+1+k_l+k_r)$, using only that
$c+1 \le \tau(r-l+2+1/\tau)$), 
and therefore we can extend $[l,r]$ to the left by up to
$\lceil 1/\tau\rceil$, or to the right by up to $\lceil 1/\tau\rceil$.

\paragraph{{(3) Segments reaching close to $A[k]$ are extended}}

The same reasoning as for the previous case applies, even if $[l,r]$ does not
originally contain position $k$. There are more restrictions, since now
$[l-k_l,r+k_r]$ must be so that it contains $k$, and the same limit
$0 \le k_l+k_r < 1+1/\tau$ applies. Thus, in addition to being possible to
extend them by at most $\lceil 1/\tau\rceil$ cells in either direction,
position $k$ must lie within the extended area.

\paragraph{Total extension} The three cases above are superimposed. Let
$\ell_l$ and $\ell_r$ the closest positions $\ell_l \le k \le \ell_r$ where
$A_x[\ell_l] = A_x[\ell_r] = 1$. Then, if $\ell_l=k$, we can set at most
$\lceil 1/\tau\rceil$ new 1s in $A_x$ to the left of $k$ by extending segments
using case
(2). Otherwise, if $k-\ell_l \le \lceil 1/\tau\rceil$, we can cover the area
$A_x[\ell_l+1,\ldots,k]$ and add up to $\lceil 1/\tau\rceil-(k-\ell_l)$
further cells to the left, using case (3). Otherwise, if $k-\ell_l > \lceil
1/\tau\rceil$, we set $\lceil 1/\tau\rceil-2$ cells to the left, apart from
$k$, using case (1). The same reasoning applies to the right, and therefore $2 \lceil
1/\tau\rceil$ is an upper bound to the number of 1s in $A_x$ produced by each
new occurrence of $x$ in $A$.
\end{proof}

The lemma shows that all the $A_x$ bitmaps add up to $O(n/\tau)$ 1s, and thus the 
lengths of all the $M_x$ bitmaps add up to $O(n/\tau)$ as well (recall that $M_x$ 
has one position per 1 in $A_x$). Therefore, we can store all the $M_x$ bitmaps
within $O(n/\tau)$ bits of space. We cannot, however, store all the $A_x$ 
bitmaps, as they may add up to $O(n^2)$ 0s (note there can be $O(n)$
distinct symbols $x$), and we still cannot probe all the $A_x$ bitmaps for all
$x$ in $o(n)$ time.

Instead, we will {\em coalesce} all the bitmaps $A_x$ into a smaller number of
bitmaps $A'_r$ (which will be called coalesced bitmaps). Coalescing works as
follows. Let us write $A[i,j]=b$ to mean $A[\ell]=b$ for all $i \le \ell \le j$.
We start with all $A'_r[1,n]=0$ for all $r$. Then we take each maximal area of 
all 1s of each bitmap, $A_x[i,j]=1$, choose some $r$ such that 
$A'_r[i-1,j+1]=0$, and set 
$A'_r[i,j]=1$. That is, we copy the run of 1s from $A_x$ to some coalesced
bitmap $A'_r$ such that the run does not overlap nor touch other previous runs
already copied (i.e., there must be at least one 0 between any two copied runs
of 1s). We associate to each such $A'_r$ a bitmap $M'_r$ where the areas
of each $M_x$ corresponding to each coalesced area of $A_x$ are concatenated,
in the same order of the coalesced areas. That is, if $A'_r[i_t,j_t]=1$,
the $t$th left-to-right run of 1s in $A'_r$, was copied from $A_x$, then 
$M_x[rank(A_x,i_t),rank(A_x,j_t)]$ will be the $t$th segment appended to $M'_r$.

\begin{example}
We can coalesce the whole bitmaps $A_1$ and $A_2$ into $A' = \langle
1~0~1~0~1~1~1 \rangle$, with the corresponding bitmap $M' = \langle 1~1~0~1~1
\rangle$.
\end{example}

The coalesced bitmaps $A'_r$ and $M'_r$ will replace the original bitmaps
$A_x$ and $M_x$.
At query time, we check for the area $[i,j]$ of each coalesced bitmap 
using Lemma~\ref{lem:algox}. 
We cannot confuse the areas of different symbols $x$ because
we force that there is at least one 0 between any two areas. We cannot report
the same $\tau'$-majority $x$ in more than one coalesced bitmap, 
as both areas should overlap on
$[i,j]$ and then they would have been merged as a single area in $A_x$.
If we find one $\tau'$-majority in one coalesced bitmap, we know that there 
is a $\tau'$-majority $x$ and can spot all of its occurrences (or the leftmost,
if desired) in optimal time, even if we cannot know the
identity of $x$. Moreover, we will find all the distinct $\tau'$-majorities in 
this way.

\subsection{Bounding the Number of Coalesced Bitmaps}

This scheme will work well if we obtain just a few coalesced bitmaps overall.
Next we show how to obtain only $O((1/\tau)\log n)$ coalesced bitmaps.

\begin{lemma} \label{lem:2lgn}
At most $2\log_{1+\tau} n$ distinct values of $x$ can have $A_x[k]=1$ for a 
given $k$.
\end{lemma}
\begin{proof}
First, $A[k]=x$ is a $\tau$-majority in $A[k,k]$, thus $A_x[k]=1$. 
Now consider any other element $x' \not=x$ such that $A_{x'}[k]=1$. This means
that $x'$ is a $\tau$-majority in some $[i,j]$ that contains $k$. Since 
$A[k]\not=x'$, it must be that $x'$ is a $\tau$-majority in $[i,k-1]$ or in 
$[k+1,j]$ (or in both). We
say $x'$ is a left-majority in the first case and a right-majority in the
second. Let us call $y_1, y_2, \ldots$ the $x'$ values that are
left-majorities, and $i_1, i_2, \ldots$ the starting points of their segments
(if they are $\tau$-majorities in several segments covering $k$, we choose one
arbitrarily). Similarly, let $z_1, z_2, \ldots$ be the $x'$ values that are 
right-majorities, and $j_1, j_2, \ldots$ the ending points of their segments.
Assume the left-majorities are sorted by decreasing values of $i_r$ and the
right-majorities are sorted by increasing values of $j_r$.
If a same value $x'$ appears in both lists, we arbitrarily remove one of them.
As an exception, we will start both lists with $y_0 = z_0 = x$, with
$i_0=j_0=k$.

It is easy to see by induction that $y_r$ must appear at least $(1+\tau)^r$ 
times in the interval $[i_r,k]$ (or in $[i_r,k-1]$, which is the same). 
This clearly holds for $y_0 = x$. Now, by the
inductive hypothesis, values $y_0, y_1, \ldots, y_{r-1}$ appear at least 
$(1+\tau)^0, (1+\tau)^1,\ldots,(1+\tau)^{r-1}$ times within $[i_{r-1},k-1]$ 
(which contains all the intervals), adding up to $\frac{(1+\tau)^r-1}{\tau}$ 
occurrences. Thus $k-1-i_{r-1}+1 \ge \frac{(1+\tau)^r-1}{\tau}$. 
In order to be a left-majority, element $y_r$ must appear strictly more than 
$\tau(k-i_{r-1}) \ge (1+\tau)^r-1$ times in $[i_r,k-1]$, to outweight all the
occurrences of the previous symbols. The case of right-majorities
is analogous.
This shows that there cannot be more than $\log_{1+\tau} n$ left-majorities 
and $\log_{1+\tau} n$ right-majorities.
\end{proof}

In the following it will be useful to define $C_x$ as the set of maximal 
contiguous 
areas of 1s in $A_x$. That is, $C_x$ is obtained by merging all the segments 
of $S_x$ that touch or overlap. 
Note that segments of $C_x$ do
not overlap, unlike those of $S_x$. Since a segment of $C_x$ covers a position 
$k$ iff some segment of $S_x$ covers position $k$ (and iff $A_x[k]=1$), it 
follows by Lemma~\ref{lem:2lgn} that any position is 
covered by at most $2\log_{1+\tau} n$ segments of $C_x$ of distinct symbols $x$.

Note that a pair of consecutive positions $A[k]=x$ and $A[k+1]=y$ is also 
covered by at most $2\log_{1+\tau} n$ such segments: the right-majorities for
$A[k]$ either are $y$ or are also right-majorities for $A[k+1]$, and those are
already among the $\log_{1+\tau} n$ right-majorities of $A[k+1]$. And vice
versa. 

We obtain $O(\log_{1+\tau} n)$ coalesced bitmaps as follows.
We take the union of all the sets $C_x$ of all the symbols $x$ and sort the
segments by their starting points. Then we
start filling coalesced bitmaps. We check if the current segment can be 
added to an existing bitmap without producing overlaps (and leaving a 0 in
between). If we can, we choose any appropriate bitmap, otherwise we start a 
new bitmap.
If at some point we need more than $2\log_{1+\tau} n$ bitmaps, it is because all
the last segments of the current $2\log_{1+\tau} n$ bitmaps overlap either the 
starting point of the current segment or the previous position, a contradiction.

\begin{example}
We have $C_1 = \{ [1,1], [5,7]\}$, $C_2 = \{ [3,3] \}$, and $C_3 = \{ [1,6] \}$.
Now, we take $C_1 \cup C_2 \cup C_3 = \{ [1,1], [1,6], [3,3],
[5,7] \}$, and the process produces precisely the coalesced bitmaps
$A'$, corresponding to the set $\{ [1,1], [3,3], [5,7] \}$, and $A_3$,
corresponding to $\{ [1,6] \}$. 
\end{example}

Note that in general the coalesced bitmaps
may not correspond to the union of complete original bitmaps $A_x$, but
areas of a bitmap $A_x$ may end up in different coalesced bitmaps.

Therefore, the coalescing process produces $O(\log_{1+\tau} n) =
O((1/\tau)\log n)$ bitmaps.
Consequently, we obtain $O((1/\tau)\log n)$ query time by simply
checking the coalesced bitmaps one by one using Lemma~\ref{lem:algox}.

Finally, representing the $O((1/\tau)\log n)$ coalesced bitmaps $A'$, which 
have total length $O((n/\tau)\log n)$ and contain $O(n/\tau)$ 1s, requires 
$O((n/\tau)\log\log n)$ bits if we use a compressed bitmap representation 
\cite{RRR07} that still offers constant-time $rank$ and $select$ queries
(recall Section~\ref{sec:bitmaps}). The coalesced bitmaps $M'$ still have
total length $O(n/\tau)$.

This completes the first part of our result. Next, we will reduce the space
usage of our encoding.

\section{Reducing the Space to $O(n\log(1/\tau))$ Bits}
\label{sec:nbits}

We introduce a different representation of the coalesced bitmaps that allows
us to store them in $O(n\log(1/\tau))$ bits, while retaining the same mechanism
described above.
We note that, although there can be $O(n/\tau)$ bits set in the bitmaps
$A_x$, each new element $x$ produces at most one new {\em run} of
contiguous 1s (case (1) in the proof of Lemma~\ref{lem:five1s}). Therefore 
there are at most $n$ runs in total. We will use a representation of coalesced 
bitmaps that takes advantage of these runs.

We will distinguish segments of $C_x$ by their lengths, separating lengths by
ranges between $\lceil 2^\ell/\tau \rceil$ and $\lceil
2^{\ell+1}/\tau\rceil -1$, for any {\em level} $0 \le \ell \le \lg(\tau n)$
(level $0$ is special in that it contains lengths starting from 1). 
In the process of creating the coalesced bitmaps described in 
the previous section, we will have separate coalesced bitmaps for inserting 
segments within each range of lengths; these will be called bitmaps of 
level $\ell$. There may be several bitmaps of the same level. It is important 
that, even with this restriction, our coalescing process will still generate 
$O((1/\tau)\log n)$ bitmaps, because only $O(1/\tau)$ coalesced bitmaps of each 
level $\ell$ will be generated.

\begin{lemma} \label{lem:8lev}
There can be at most $4/\tau$ segments of any $C_x$, of length between 
$\lceil 2^\ell/\tau \rceil$ and $\lceil 2^{\ell+1}/\tau \rceil -1$, 
covering a given position $k$, for any $\ell$.
\end{lemma}
\begin{proof}
Any such segment must be contained in the area $A[k-\lceil
2^{\ell+1}/\tau\rceil+1, k+\lceil 2^{\ell+1}/\tau\rceil-1]$,
and if $x$ is a $\tau$-majority in it, it must appear more than 
$\tau \lceil 2^\ell / \tau \rceil \ge 2^\ell $ times.
There can be at most $4/\tau$ different values of $x$ appearing more than 
$2^\ell$ times in an area of length less than $2^{\ell+2}/\tau$.
\end{proof}

Consider a coalesced bitmap $A'[1,n]$ of level $\ell$. All of its 1s come in
runs of lengths at least $b = \lceil 2^\ell/\tau\rceil$. We cut $A'$ into
{\em chunks} of length $b$ and define two bitmaps: $A'_1[1,n/b]$ will have $A'_1[i]=1$
iff the $i$th chunk of $A'$ is all 1s, and $A'_2[1,n/b]$ will have $A'_2[i]=1$
iff the $i$th chunk of $A'$ has 0s and 1s. Note that, since the runs of 1s are
of length at least $b$, inside a chunk with 0s and 1s there can be at most one 
01 and at most one 10, and the 10 can only come before the 01. 
Let $p_{10}[j]$ be the position, in the $j$th chunk
with 0s and 1s, of the 1 preceding a 0, where $p_{10}[j]=0$ if the chunk starts
with a 0. Similarly, let $p_{01}[j]$ be the position of the 0 preceding a 1,
with $p_{01}[j]=b$ if the chunk ends with a 0. It always holds that $p_{10}[j] <
p_{01}[j]$, and the number of 1s in the chunk is $r(j)=p_{10}[j]+(b-p_{01}[j])$.
Also, the rank up to position $k$ in the chunk, $r(j,k)$, is $k$ if
$k \le p_{10}[j]$, $p_{10}[j]$ if $p_{10}[j] < k \le p_{01}[j]$, and
$p_{10}[j]+(k-p_{01}[j])$ if $k > p_{01}[j]$. Then it holds that
\begin{eqnarray*}
&& rank(A',i) ~~=~~ b\cdot r_1 ~+~ 
                 \sum_{j=1}^{r_2} r(j) ~+~ \\
&& ~~~~~~~~~~~~~~~ 
 [\mathbf{if}~A'_2[1+\lfloor i/b \rfloor]=1~\mathbf{then}~
          r(r_2+1,k)~\mathbf{else}~
		A'_1[1+\lfloor i/b \rfloor]\cdot k],
\end{eqnarray*}
where $r_1 = rank(A'_1,\lfloor i/b\rfloor)$,
$r_2 = rank(A'_2,\lfloor i/b\rfloor)$, and $k = i~\mathrm{mod}~b$.
Note this can be computed in constant time as long as we have constant-time
$rank$ data structures on $A'_1$ and $A'_2$, and constant-time access and sums
on $p_{10}$ and $p_{01}$.

\begin{example}
Using $b=2^\ell$ to make it more interesting,
we would have three coalesced bitmaps: $A' = \langle
1~0~1~0~0~0~0 \rangle$, of level $\ell=0$, for the segments $[1,1]$ and
$[3,3]$;
$A'' = \langle 0~0~0~0~1~1~1 \rangle$, of level $\ell=1$, for the segment
$[5,7]$; and $A''' = \langle 1~1~1~1~1~1~0 \rangle$, of level $\ell=2$, for the
segment $[1,6]$. Consider level $\ell=0$ and $b=2$, and let us focus on $A'$. 
Then, we would have
$A'_1 = \langle 0~0~0~0 \rangle$,
$A'_2 = \langle 1~1~0~0 \rangle$, 
$p_{10} = \langle 1~1 \rangle$, and
$p_{01} = \langle 2~2 \rangle$.
\end{example}

To have constant-time sums on $p_{10}$ ($p_{01}$ is analogous), we store its
values in a bitmap $A'_{10}$, where we set all the bits at positions
$r + \sum_{j=1}^r p_{10}[j]$ to 1, for all $r$. Then we can recover 
$\sum_{j=1}^r p_{10}[j] = select(A'_{10},r)-r$. We use a bitmap representation
\cite{OS07} that solves $select$ in constant time (recall
Section~\ref{sec:bitmaps}).
Let $n'$ be the number of segments $C_x$ represented in bitmap
$A'$. Then there are at most $2n'$ chunks with 0s and 1s, and $A'_{10}$ contains
at most $2n'$ 1s and $2n'b$ 0s (as $0 \le p_{10}[j] \le b$). 
The size of the bitmap representation \cite{OS07} is
in this case $O(n'\log b) = O(n'(\ell + \log(1/\tau)))$ bits. On the other
hand, bitmaps $A'_1$ and $A'_2$ are represented in plain form
\cite{Mun96,Cla96}, requiring $O(n/b) = O(n\tau/2^\ell)$ bits.

Considering that there are $O(n/\tau)$ 1s overall, and that the runs of level
$\ell$ are of length at least $2^\ell/\tau$, we have that there can be at
most $n/2^\ell$ runs across the $O(1/\tau)$ bitmaps of level $\ell$.
Therefore, adding up the space over the bitmaps of level $\ell$, we have 
$O(n(\ell+\log(1/\tau))/2^\ell)$ bits. Added over all the levels $\ell$, this 
gives $O(n\log(1/\tau))$ bits.

Let us now consider the representation of the coalesced bitmaps $M'$. They
have total length $O(n/\tau)$ and contain $n$ 1s overall, therefore using the
representation of Lemma~\ref{lem:bitmap} with $s=1$, we have 
$O(n\log(1/\tau))$ bits of space. They solve $rank$ queries in time
$O(\log\log_w(1/\tau))$, and $select$ in constant time.

As we have to probe $O((1/\tau)\log n)$ coalesced bitmaps $M'$ in the worst
case, this raises our query time to $O((1/\tau)\log\log_w(1/\tau)\log n)$. 
This concludes the proof of Theorem~\ref{thm:main}, except for the
construction time (see the next section).

In our previous work \cite{NT14}, we had obtained
$O((1/\tau)\log n)$ time, but using $O((n/\tau)\log^*n)$ bits of space. It is 
not hard to obtain that time, using $O(n/\tau)$ bits, by simply representing
the coalesced bitmaps $M'$ using plain $rank$/$select$ structures
\cite{Cla96,Mun96}, or even using $O(n\log(1/\tau) +
(n/\tau)/\polylog n)$ bits, for any $\polylog n$, using compressed
representations \cite{Pat08}. The extra $O(\log\log_w(1/\tau))$ time factor
arises when we insist in obtaining the optimal $O(n\log(1/\tau))$ bit space.
We note that this time penalty factor vanishes when $1/\tau = w^{O(1)}$, which
includes the case where $1/\tau$ grows
polylogarithmically with $n$.

\section{Construction} \label{sec:constr}

The most complex part of the construction of our encoding is to build the
sets $C_x$. Once these are built, the structures described in
Section~\ref{sec:nbits} can be easily constructed in $o(n\log n)$ time:

\begin{enumerate}
\item The $O(n)$ segments $C_x$ belong to $[1,n]$, so they are sorted by 
starting point in $O(n)$ time.
\item We maintain a priority queue for each level $\ell$, containing the last 
segment of each coalesced bitmap. We use the queue to find the segment that 
finishes earliest in order to try to add the new segment of $C_x$ after it. 
We carry out, in total, $O(n)$ operations on those queues, and each contains 
$O(1/\tau)$ elements, thus they take total time $O(n\log(1/\tau))=o(n\log n)$.
\item The bitmaps $A'$ of each level $\ell$, represented with $A'_1$, $A'_2$,
$A'_{01}$ and $A'_{10}$, are easily built in $O(n/b)=O(n\tau/2^\ell)$ time.
Added over the $O(1/\tau)$ coalesced bitmaps of level $\ell$ this is
$O(n/2^\ell)$, and added over all the levels $\ell$ this gives $O(n)$ total
time.
\item The coalesced bitmaps $M'$ have $O(n)$ 1s overall, so their
representation (Lemma~\ref{lem:bitmap}) is also built in $O(n)$ time, except
for the predecessor structures, which need construction of deterministic
dictionaries. This can be done in $o(n\log n)$ total time \cite{Ruz08}.
\end{enumerate}

Now we show that the sets $C_x$ can be built in $O(n\log n)$ time, thus 
finishing the proof of Theorem~\ref{thm:main}.

We build the set of increasing positions $P_x$ where $x$ appears in $A$,
for each $x$, in $O(n\log n)$ total time (the elements of $A$ can be 
of any atomic type, so we only rely on a comparison-based dictionary to
maintain the set of different $x$ values and their $P_x$ lists). Now we
build $C_x$ from each $P_x$ using a divide-and-conquer approach, in
$O(|P_x|\log|P_x|)$ time, for a total construction time of $O(n\log n)$.

We pick the middle element $k \in P_x$ and compute in linear time the segment 
$[l,r] \in C_x$ that contains $k$. To compute $l$, we find the leftmost 
element $p_l \in P_x$ such that $x$ is a $\tau$-majority in $[p_l,k_r]$, for some
$k_r \in P_x$ with $k_r \ge k$. 

To find $p_l$, we note that it must hold that $(w(p_l,k-1)+w(k,k_r))/(k_r-p_l+1) >
\tau$, where $w(i,j)$ is the number of occurrences of $x$ in $A[i,j]$. The
condition is equivalent to $w(p_l,k-1)/\tau + p_l - 1 > k_r - w(k,k_r)/\tau$. 
Thus we compute in linear time the minimum value $v$ of $k_r - w(k,k_r)/\tau$ 
over all those $k_r \in P_x$ to the right of $k$, and then traverse all those 
$p_l \in P_x$ to the left of $k$, left to right, to find the first one that 
satisfies $w(p_l,k-1)/\tau + p_l + 1 > v$, also in linear time. 
Once we find the proper $p_l$ and its corresponding $k_r$, the starting 
position of the segment is slightly adjusted to the left of $p_l$, to be the 
smallest value that satisfies $w(p_l,k_r)/(k_r-l+1) > \tau$, that is, $l$
satisfies $l > -w(p_l,k_r)/\tau +k_r+1$, or $l = k_r - \lceil w(p_l,k_r)/\tau
\rceil + 2$.

Once $p_r$ and then $r$ are computed analogously, we insert $[l,r]$ into $C_x$
and continue recursively with the elements of $P_x$ to the left of $p_l$ and
to the right of $p_r$. Upon return, it might be necessary to join $[l,r]$ with
the rightmost segment of the left part and/or with the leftmost segment of the
right part, in constant time. The total construction time is $T(n) = O(n) + 
2T(n/2) = O(n\log n)$. 

\paragraph{Building multiple structures}

In order to answer $\tau'$-majority queries for any $\tau \le \tau' < 1$ in 
time related to $1/\tau'$ and not to $1/\tau$, we build the encoding
of Theorem~\ref{thm:main} for values $\tau'' = 1/2, 1/4, 1/8, \ldots, 
1/2^{\lceil \lg 1/\tau \rceil}$. Then, a $\tau'$-majority query is run on the 
structure built for $\tau'' = 1/2^{\lceil \lg 1/\tau' \rceil}$. Since 
$\tau'/2 < \tau'' \le \tau'$, the query time is 
$O((1/\tau'')\log\log_w(1/\tau'')\log
n)=O((1/\tau')\log\log_w(1/\tau')\log n)$.

As for the space, we build $O(\log(1/\tau))$ structures, so we use
$O(n\log^2(1/\tau))$ bits, and the construction time is
$O(n\log(1/\tau)\log n)$.

\begin{corollary} \label{cor:vartau}
Given a real number $0<\tau<1$, there exists an encoding using 
$O(n\log^2(1/\tau))$ bits that answers range $\tau'$-majority queries,
for any $\tau \le \tau' < 1$, in time $O((1/\tau')\log\log_w(1/\tau')\log
n)$, where $w=\Omega(\log n)$ is the RAM word size in bits.
The structure can be built in time $O(n\log(1/\tau)\log n)$.
\end{corollary}

\section{A Faster Data Structure}
\label{sec:optimal}

In this section we show how, by adding $O(n\log\log n)$ bits to our data 
structure, we can slash a $\log n$ factor from the query time, that is, we prove
Theorem~\ref{thm:optimal}. The result, as discussed in the Introduction,
yields the optimal query time $O(1/\tau)$ when $1/\tau = O(\polylog n)$, 
although the resulting space may not be optimal anymore.

The idea is inspired in a previous non-encoding data structure for majority
queries \cite{wads}. Consider a value $\ell$. Then we will cut $A$ into
consecutive pieces of length $2^\ell$ (said to be of {\em level} $\ell$) in 
two overlapped ways: $A[2^\ell k+1,2^\ell(k+1)]$ and $A[2^\ell
k+2^{\ell-1}+1,2^\ell(k+1)+2^{\ell-1}]$, for all $k \ge 0$.
We carry out this partitioning for every $\lceil \lg(1/\tau) \rceil \le \ell
\le \lceil \lg n \rceil$.

Note that there are $O(n/2^\ell)$ pieces of level $\ell$, and any interval 
$A[i,j]$ of length up to $2^\ell/2$ is contained in some piece $P$ of level 
$\ell$. Now, given a query interval $A[i,j]$, let
$\ell = \lceil \lg(j-i+1)\rceil+1$. Then, not only
$A[i,j]$ is contained in a piece $P$ of level $\ell$, but also any 
$\tau$-majority $x$ in $A[i,j]$ must be a $\tau/4$-majority in $P$: 
Since $j-i+1 > 2^\ell/4$, $x$ occurs more than $\tau(j-i+1) > (\tau/4)2^\ell$ 
times in $A[i,j]$, and thus in $P$.

Consider a $\tau/4$-majority $x$ in a given piece $P$ of level $\ell$ that is 
also a $\tau$-majority for some range $A[i,j]$ within $P$, where $2^\ell/4 <
j-i+1 \le 2^\ell/2$. By construction of our previous structures, there exists 
a maximal segment $C_x$ that contains the range $[i,j]$. If there is 
another range $A[i',j']$ within $P$ where $x$ is a $\tau$-majority, then there
exists another maximal segment $C'_x$ for the same $x$ within $P$. By our 
construction, if $C'_x \neq C_x$, then $C'_x$ is disjoint with $C_x$,
and thus each of them contains at least $(\tau/4)2^\ell$ distinct occurrences of
$x$. Obviously, segments $C_y$ for $\tau$-majorities $y \neq x$ contain
other $(\tau/4)2^\ell$ occurrences disjoint from those of $x$. Therefore, the 
number of distinct maximal segments $C$ that contain $\tau$-majorities at any
range $A[i,j]$ (with $j-i+1 > 2^\ell/4$) within $P$ is upper bounded by 
$4/\tau$. We will say those segments $C$ are {\em relevant} to $P$.

Therefore, for each piece $P$ of level $\ell$, we will store the index $r$ of
the coalesced
bitmap $A'_r$ (and its companion $M'_r$) to which each maximal segment $C$ that 
is relevant to $P$ belongs. Since there are at most $4/\tau$ such coalesced
bitmaps to record, out of a total of $O((1/\tau)\log n)$ coalesced bitmaps, 
$\gamma$-codes on a differential encoding of the subset values requires
$O((1/\tau)\log\log n)$ bits.\footnote{We could also afford to store
them in plain form using $O((1/\tau)(\log(1/\tau)+\log\log n))$ bits.} Added 
up over the $O(n/2^\ell)$ pieces of level $\ell \ge \lceil \lg(1/\tau)\rceil$, 
this yields $\sum_{\ell \ge \lceil\lg(1/\tau)\rceil} O((n/2^\ell)
(1/\tau) \log\log n) = O(n\log\log n)$ bits.

This information reduces the search effort to that of verifying $O(1/\tau)$
coalesced bitmaps $A'_r$ and $M'_r$ for the range $[i,j]$, and thus to
$O((1/\tau)\log\log_w (1/\tau))$ query time. However, for ranges shorter than
$1/\tau$, where no piece structure has been built, we still have the original
query time. To speed up this case, we build a second structure where, for each 
element $A[k]$, we identify the coalesced bitmap where the maximal segment 
$C_{A[k]}$ containing the segment $A[k,k]$ belongs, and store the identifier $r$
of the corresponding coalesced bitmap $A'_r$ (and $M'_r$) associated to $k$. This 
requires $O(n\log((1/\tau)\log n)) = O(n\log(1/\tau) + n\log\log n)$ further 
bits, and allows checking only one coalesced bitmap $A'_r$ (and $M'_r$) for each
of the $O(1/\tau)$ positions that need to be checked.

To finish the proof we must consider the construction time. The second
structure (for short ranges) is easily built with the general structure,
taking no additional time, by keeping track of which maximal segment 
$C_{A[k]}$ contains each segment $A[k,k]$ and which coalesced bitmap it is 
assigned. With this, the structure for long ranges can be built as follows:
for each position $A[k]$ contained in a piece $P$ of level $\ell$, consider 
the maximal 
segment $C_{A[k]}$ that contains it and determine whether it is relevant to $P$.
A weak test for this is to consider the coalesced bitmap $M'$ where
$C_{A[k]}$ is represented (which is precisely what the first structure stores
associated to $k$) and ask whether $M'$ contains more than $(\tau/4)2^\ell$
1s in the range of $P$. 
This must be the case if $C_{A[k]}$ is relevant to $P$. Although including
the identifier of each $M'$ that passes the test may add some nonrelevant
ones, we still cannot include more than $4/\tau$ coalesced bitmaps in the set,
as the 1s in the $M'$ bitmaps are disjoint.

The $rank$ operations on bitmaps $M'$ take $O(\log\log_w(1/\tau))$ time,
so we avoid them to count how many 1s does $M'$ contain in the range of $P$. 
Instead, we perform a preprocessing pass over $P$ as
follows: We initialize to zero a set of $O((1/\tau)\log n)$ counters, one per 
coalesced bitmap $M'$, and process $P$ left to right. We increase the counter 
associated to the bitmap $M'$ of each element $A[k]$ in $P$. At the end, we
know all the desired values. This takes $O(2^\ell)$ time, and a similar
postprocessing pass clears the counter for the next piece.

Therefore, we process all the pieces $P$ of level $\ell$ in time $O(2^\ell)$, 
which amounts to $O(n)$ time per level. Added over all the levels, this gives 
$O(n\log n)$ total time. This concludes the proof of Theorem~\ref{thm:optimal}.

\section{Conclusions}
\label{sec:concl}

A $\tau$-majority query on array $A[1,n]$ receives a range $[i,j]$ and 
returns all the elements appearing more than $\tau(j-i+1)$ times in $A[i,j]$. 
We have obtained the first results about {\em encodings} for answering range 
$\tau$-majority queries. Encodings are data structures that use less space than
what is required to store $A$ and answer queries without accessing $A$ at all.
In the encoding scenario we do not
report the $\tau$-majorities themselves, but one of their positions in $A[i,j]$.

We have proved that $\Omega(n\log(1/\tau))$ bits are necessary for any such 
encoding, even if it can only count the number of $\tau$-majorities in
any range. Then we presented an encoding that uses the optimal 
$O(n\log(1/\tau))$ bits, and answers queries in 
$O((1/\tau)\log\log_w(1/\tau)\log n)$ time in the RAM model with word size
$w=\Omega(\log n)$ bits. We also showed that this time can be divided by
$\log n$ if we add $O(n\log\log n)$ bits to the space. This yields various
space/time tradeoffs, shown in Table~\ref{tab:final}.
Our encoding can actually report any occurrence of each $\tau$-majority, in 
optimal extra time. The structure is built in $O(n\log n)$ time.

An open question is whether it is possible to achieve optimal query time within
optimal space for all values of $1/\tau$. As seen in Table~\ref{tab:final}, we
reach this only for $\log(1/\tau)=\Theta(\log\log n)$. This is also possible 
when $\log (1/\tau) = \Omega(\log n)$, where we leave the non-encoding scenario
\cite{wads}. Instead, our results for $\log(1/\tau)$ between $\log\log n$ and
$\log n$ have a small factor $O(\log\log_w (1/\tau))$ over the optimal time,
and those for $\log(1/\tau)$ below $\log\log n$ either require nonoptimal
$O(n\log\log n)$ bits of space, or an $O(\log n)$ factor over the optimal
time. It is not clear whether combined optimality can be reached.

Another open question is whether we can do better for weaker versions of the
problem we have not studied. For example, if we are only required to report
{\em any} occurrence of {\em any} $\tau$-majority (or, even less, telling 
whether or not there exists a $\tau$-majority), our lower bound based on
representing a bitmap $B$ shows that $\Omega(n)$ bits are necessary, but we
do not know if this bound is tight.


\bibliographystyle{elsarticle-harv}
\bibliography{paper}

\begin{thebibliography}{27}
\expandafter\ifx\csname natexlab\endcsname\relax\def\natexlab#1{#1}\fi
\expandafter\ifx\csname url\endcsname\relax
  \def\url#1{\texttt{#1}}\fi
\expandafter\ifx\csname urlprefix\endcsname\relax\def\urlprefix{URL }\fi

\bibitem[{Belazzougui et~al.(2013)Belazzougui, Gagie, and Navarro}]{wads}
Belazzougui, D., Gagie, T., Navarro, G., 2013. Better space bounds for
  parameterized range majority and minority. In: Proc. 11th Annual Workshop on
  Algorithms and Data Structures (WADS). pp. 121--132.

\bibitem[{Berkman and Vishkin(1993)}]{BV93}
Berkman, O., Vishkin, U., 1993. Recursive star-tree parallel data structure.
  SIAM Journal on Computing 22~(2), 221--242.

\bibitem[{Bose et~al.(2005)Bose, Kranakis, Morin, and Tang}]{lboo}
Bose, P., Kranakis, E., Morin, P., Tang, Y., 2005. Approximate range mode and
  range median queries. In: Proc. 22nd International Symposium on Theoretical
  Aspects of Computer Science (STACS). pp. 377--388.

\bibitem[{Brodal et~al.(2009)Brodal, Fagerberg, Greve, and
  L{\'o}pez-Ortiz}]{bro}
Brodal, G., Fagerberg, R., Greve, M., L{\'o}pez-Ortiz, A., 2009. Online sorted
  range reporting. In: Proc. 20th Annual International Symposium on Algorithms
  and Computation (ISAAC). pp. 173--182.

\bibitem[{Chan et~al.(2012{\natexlab{a}})Chan, Durocher, Larsen, Morrison, and
  Wilkinson}]{stacs}
Chan, T., Durocher, S., Larsen, K., Morrison, J., Wilkinson, B.,
  2012{\natexlab{a}}. Linear-space data structures for range mode query in
  arrays. In: Proc. 29th International Symposium on Theoretical Aspects of
  Computer Science (STACS). pp. 290--301.

\bibitem[{Chan et~al.(2012{\natexlab{b}})Chan, Durocher, Skala, and
  Wilkinson}]{swat}
Chan, T., Durocher, S., Skala, M., Wilkinson, B., 2012{\natexlab{b}}.
  Linear-space data structures for range minority query in arrays. In: Proc.
  13th Scandinavian Symposium on Algorithmic Theory (SWAT). pp. 295--306.

\bibitem[{Chan and Wilkinson(2013)}]{CW13}
Chan, T., Wilkinson, B., 2013. Adaptive and approximate orthogonal range
  counting. In: Proc. 24th Annual ACM-SIAM Symposium on Discrete Algorithms
  (SODA). pp. 241--251.

\bibitem[{Clark(1996)}]{Cla96}
Clark, D., 1996. Compact {PAT} trees. Ph.D. thesis, University of Waterloo,
  Canada.

\bibitem[{Durocher et~al.(2013)Durocher, He, Munro, Nicholson, and
  Skala}]{Steph}
Durocher, S., He, M., Munro, I., Nicholson, P., Skala, M., 2013. Range majority
  in constant time and linear space. Information and Computation 222, 169--179.

\bibitem[{Elias(1974)}]{Eli74}
Elias, P., 1974. Efficient storage and retrieval by content and address of
  static files. Journal of the ACM 21, 246--260.

\bibitem[{Fano(1971)}]{Fan71}
Fano, R., 1971. On the number of bits required to implement an associative
  memory. Memo 61, Computer Structures Group, Project MAC, Massachusetts.

\bibitem[{Fischer and Heun(2011)}]{RMQ1}
Fischer, J., Heun, V., 2011. Space-efficient preprocessing schemes for range
  minimum queries on static arrays. SIAM Journal of Computing 40~(2), 465--492.

\bibitem[{Gagie et~al.(2011)Gagie, He, Munro, and Nicholson}]{spire}
Gagie, T., He, M., Munro, I., Nicholson, P., 2011. Finding frequent elements in
  compressed 2d arrays and strings. In: Proc. 18th International Symposium on
  String Processing and Information Retrieval (SPIRE). pp. 295--300.

\bibitem[{Greve et~al.(2010)Greve, J{\o}rgensen, Larsen, and Truelsen}]{lbo}
Greve, M., J{\o}rgensen, A., Larsen, K.~D., Truelsen, J., 2010. Cell probe
  lower bounds and approximations for range mode. In: Proc. 37th International
  Colloquium on Automata, Languages and Programming (ICALP). pp. 605--616.

\bibitem[{Grossi et~al.(2013)Grossi, Iacono, Navarro, Raman, and
  Satti}]{grossi}
Grossi, R., Iacono, J., Navarro, G., Raman, R., Satti, S.~R., 2013. Encodings
  for range selection and top-k queries. In: Proc. 21st Annual European
  Symposium on Algorithms (ESA). pp. 553--564.

\bibitem[{Karpinski and Nekrich(2008)}]{KarpinskiN08}
Karpinski, M., Nekrich, Y., 2008. Searching for frequent colors in rectangles.
  In: Proc. 20th Canadian Conference on Computational Geometry (CCCG). pp.
  11--14.

\bibitem[{Karpinski and Nekrich(2011)}]{NNN}
Karpinski, M., Nekrich, Y., 2011. Top-k color queries for document retrieval.
  In: Proc. 22nd Annual ACM-SIAM Symposium on Discrete Algorithms (SODA). pp.
  401--411.

\bibitem[{Munro(1996)}]{Mun96}
Munro, I., 1996. Tables. In: Proc. 16th Conference on Foundations of Software
  Technology and Theoretical Computer Science (FSTTCS). pp. 37--42.

\bibitem[{Navarro et~al.(2014)Navarro, Raman, and Rao}]{NRR14}
Navarro, G., Raman, R., Rao, S.~S., 2014. Asymptotically optimal encodings for
  range selection. In: Proc. 34th Annual Conference on Foundations of Software
  Technology and Theoretical Computer Science (FSTTCS). LNCS. To appear.

\bibitem[{Navarro and Thankachan(2014)}]{NT14}
Navarro, G., Thankachan, S., 2014. Encodings for range majority queries. In:
  Proc. 25th Annual Symposium on Combinatorial Pattern Matching (CPM). LNCS
  8486. pp. 262--272.

\bibitem[{Okanohara and Sadakane(2007)}]{OS07}
Okanohara, D., Sadakane, K., 2007. Practical entropy-compressed rank/select
  dictionary. In: Proc. 9th Workshop on Algorithm Engineering and Experiments
  (ALENEX). pp. 60--70.

\bibitem[{Petersen and Grabowski(2009)}]{R2}
Petersen, H., Grabowski, S., 2009. Range mode and range median queries in
  constant time and sub-quadratic space. Information Processing Letters
  109~(4), 225--228.

\bibitem[{P\u{a}tra\c{s}cu and Thorup(2008)}]{PT08}
P\u{a}tra\c{s}cu, M., Thorup, M., 2008. Time-space trade-offs for predecessor
  search. CoRR cs/0603043v1, {\tt http://arxiv.org/pdf/cs/0603043v1}.

\bibitem[{P\v{a}tra\c{s}cu(2008)}]{Pat08}
P\v{a}tra\c{s}cu, M., 2008. Succincter. In: Proc. 49th Annual IEEE Symposium on
  Foundations of Computer Science (FOCS). pp. 305--313.

\bibitem[{Raman et~al.(2007)Raman, Raman, and Rao}]{RRR07}
Raman, R., Raman, V., Rao, S.~S., 2007. Succinct indexable dictionaries with
  applications to encoding {\it k}-ary trees, prefix sums and multisets. ACM
  Transactions on Algorithms 3~(4), article 43.

\bibitem[{Ru{\v{z}}i{\'c}(2008)}]{Ruz08}
Ru{\v{z}}i{\'c}, M., 2008. Constructing efficient dictionaries in close to
  sorting time. In: Proc. 35th International Colloquium on Automata, Languages
  and Programming (ICALP). LNCS 5125. pp. 84--95 (part I).

\bibitem[{Skala(2013)}]{skala}
Skala, M., 2013. Array range queries. In: Space-Efficient Data Structures,
  Streams, and Algorithms. LNCS. Springer, pp. 333--350.

\end{thebibliography}

\end{document}